\documentclass[letterletter, 10 pt, journal, twoside]{ieeetran}
\ifCLASSINFOpdf
\else
\fi
\usepackage[comma,numbers,square,sort&compress]{natbib}
\usepackage{algorithm} %format of the algorithm
\usepackage{algorithmic} %format of the algorithm
\usepackage{setspace}
\usepackage{arydshln}
\usepackage{multirow} %multirow for format of table
\usepackage{xcolor}
\usepackage{setspace}
\usepackage{amsmath}
\newtheorem{theorem}{\textbf{Theorem}}
\newtheorem{lemma}{\textbf{Lemma}}
\newtheorem{example}{\textbf{Example}}

\newtheorem{remark}{\textbf{Remark}}
\newtheorem{definition}{\textbf{Definition}}

\newtheorem{proposition}{\textbf{Proposition}}

\newcommand\myeq{\mathrel{\stackrel{\makebox[0pt]{\mbox{\normalfont\tiny (b)}}}{=}}}
\newcommand\mge{\mathrel{\stackrel{\makebox[0pt]{\mbox{\normalfont\tiny (a)}}}{\ge}}}
\newcommand\myeqa{\mathrel{\stackrel{\makebox[0pt]{\mbox{\normalfont\tiny (a)}}}{=}}}

\usepackage{pmat}
\newenvironment{proof}{{\noindent{\bf \emph{Proof:}}}\quad}{\hfill $\square$\par}

\usepackage{multirow}
\usepackage{url}
\usepackage{subfigure}
\usepackage{fancyhdr}
\usepackage{amsmath}
\usepackage{multirow}
\usepackage{amssymb }
\usepackage{color}
\usepackage{graphics} % for pdf,  bitmapped graphics files
\usepackage{graphicx}
\usepackage{geometry}
\usepackage{setspace}
\begin{document}
%\begin{spacing}{1.5}
%
% letter title
%\title{\LARGE \bf
%
%Structural Controllability of Linear Descriptor Networks with Local Symmetries
%%On Link Selection and Controllability Robustness for Networks: Complexity Analysis
%%On Link Selection and Robustness for Network Controllability: Complexity Analysis
%Perturbation-Tolerant Structural Controllability with Application to the Controllability Radius Problems

%
\title{\LARGE {PTSC: a New Definition for Structural Controllability \\~ under Numerical Perturbations}} %On a New Strong Structural Controllability Definition with Implications to the Controllability Radius Problems
%\title{Refined Study on Strong Structural Controllability of Linear Systems via a Randomized Approach} %Generic Detectability and Isolability of Topology Failures in Networks of Linear Systems%  Networked Dynamic Systems
% On a new strong structural controllability for linear systems and its applications to the controllability radius problems

\author{Yuan Zhang, Yuanqing Xia
\thanks{{\bf This paper is a {\emph{full}} version of a conference paper to appear in the 40th Chinese Control Conference on July 26-28, 2021 in Shanghai, China.}.

 This work was supported in part by the  China Postdoctoral Innovative Talent Support Program (BX20200055), the China Postdoctoral
Science Foundation (2020M680016), the National Natural Science Foundation of China (62003042), and the State Key Program of National Natural Science Foundation of China (61836001).

Yuan Zhang and Yuanqing Xia are with the School of Automation, Beijing Institute of Technology, Beijing, China {(email: {\tt\small zhangyuan14@bit.edu.cn, xia\_yuanqing@bit.edu.cn}).} %(emails: zhangyuan14@bit.edu.cn, xia_yuanqing@bit.edu.cn,	kunliubit@bit.edu.cn)

%Jun Shang is with the Department of Electrical and Computer Engineering, University of Alberta, Edmonton, Canada T6G 1H9 (email:{\tt\small  jshang2@ualberta.ca}).
       } }
\pagestyle{empty} % Removes all the page numbers (except for the title page)
\maketitle
\thispagestyle{empty} % Removes the page number in the first page
% in the abstract or keywords. %malicious
%\geometry{left=1.69cm,right=1.69cm,top=2.10cm,bottom=1.51cm}
\begin{abstract}
This paper proposes a novel notion for structural controllability under structured numerical perturbations, namely the perturbation-tolerant structural controllability (PTSC), on a single-input structured system whose entries can be classified into three categories: fixed zero entries, unknown generic entries whose values are fixed but unknown, and perturbed entries that can take arbitrary complex values. Such a system is PTSC if, for almost all values of the unknown generic entries in the parameter space, the corresponding controllable system realizations can preserve controllability under arbitrary complex-valued perturbations with their structure prescribed by the perturbed entries. This new notion can characterize the generic property in controllability preservation under structured numerical perturbations. We give a necessary and sufficient condition for a single-input system to be PTSC, whose verification has polynomial time complexity. Our results can serve as some feasibility conditions for the conventional {\emph{structured controllability radius}} problems from a generic view.
\end{abstract}
%\overrideIEEEmargins
% Comment this command for final version.
\section{Introduction}
In recent years, security has been becoming an attractive issue in the control and estimation of cyber-physical systems, such
as chemical processes, power grids and transportation networks \cite{wood2002denial,Buldyrev2009Catastrophic,fawzi2014secure,mitra2019byzantine}. The robustness of various system properties have been investigated under internal faults (like disconnections of links/nodes \cite{Buldyrev2009Catastrophic,zhang2020generic}) or external attacks (like adversarial sensor/actuator attacks \cite{fawzi2014secure}), including stability \cite{FabioFragility2018}, stabilization \cite{De2015Input}, controllability and observability \cite{commault2008observability,Rahimian2013Structural,zhang2019minimal,Y_Zhang_2017,zhang2018arxiv}. Particularly, as a fundamental system property, controllability/observability under structural perturbations (such as link/node/actuator/sensor removals or deletions) has been extensively explored on its robust performance. To name a few, \cite{commault2008observability} considered observability preservation under sensor removals, \cite{Rahimian2013Structural} investigated controllability preservation under simultaneous link and node failures, while \cite{zhang2019minimal,Y_Zhang_2017,zhang2018arxiv} systematically studied the involved optimization problems with respect to link/node/actuator/sensor removals from a computational perspective. Since controllability/observability is a generic property that depends mainly on the system structure \cite{generic}, its robustness is mainly dominated by the robustness property of the corresponding graphs.

Note that structural perturbation is a kind of numerical perturbation that makes the corresponding link has a zero weight. In the more general case where the perturbed links do not necessarily result in zero weights, controllability robustness has also been studied by computing the distance (in terms of the $2$-norm or the Frobenius norm) from a controllable system to the set of uncontrollable systems \cite{R.E1984Between,WickDistance1991,Hu2004Real,Gu2006distance}. Such a notion, also named {\emph{controllability radius}}, was first proposed by \cite{R.E1984Between}, and then developed by several other researchers on its efficient computations \cite{WickDistance1991,Hu2004Real,Gu2006distance}. Recently, by restricting the perturbation matrices to a prescribed structure, the so-called {\emph{structured controllability radius problem}} (SCRP) has also attracted researchers' interest, i.e., the problem of determining the smallest (Frbenius or $2$-) norm additive perturbation with a prescribed structure for which controllability fails to hold \cite{KarowStructured2009}. Towards this problem, various numerical algorithms have been proposed \cite{khare2012computing,johnson2018structured,bianchin2016observability}. However, due to the nonconvexity of this problem, all these algorithms are suboptimal \cite{johnson2018structured}. Moreover, since most of these algorithms adopted some relaxation and iterative techniques and owing to the involved rounding errors, there is usually no guarantee that the returned numerical perturbations can make the original system uncontrollable.

On the other hand, to avoid the potential numerical issues, strong structural controllability (SSC), a notion proposed by Mayeda and Yamada \cite{mayeda1979strong}, could also be used to measure the controllability robustness of a system against numerical perturbations. In the SSC theory, the system parameters are divided into two categories, indeterminate parameters and fixed zero parameters. A system is SSC, if whatever values (other than zero) the indeterminate parameters of the system may take, the system is controllable. Criteria for SSC in the single-input case was given in \cite{mayeda1979strong}, and then extended to the multi-input cases in \cite{bowden2012strong,monshizadeh2014zero}, as well as  allowing the existence of parameters that can take arbitrary values, including zero and nonzero \cite{jia2020unifying}. Note for SSC to measure controllability robustness, the numerical perturbations should have the same zero/nonzero structure as the original systems. While in practice, perturbations could be imposed to only partial system components (such as a subset of links of a network) and do not necessarily have the same structure as the original systems.

In this paper, we propose a new definition for structural controllability under structured numerical perturbations, namely, the perturbation-tolerant structural controllability (PTSC). The entries for system matrices are classified into three categories: fixed zero entries, unknown generic entries whose values are fixed but unknown (they can be seen as randomly selected values), and perturbed entries that can take arbitrary complex values. The notion of PTSC is defined in the following way: a system is PTSC if, for almost all values of the unknown generic entries in the parameter space, the corresponding controllable system realizations can preserve controllability under arbitrary complex-valued perturbations with structure prescribed by the perturbed entries. The main contributions of this paper are as follows:
 %Unlike the SSC, PTSC is a generic property in the sense that, depending on the system structure, either for almost all controllable system numerical realizations, there exists a structured numerical perturbation prescribed by the perturbed entries so that the resulting system is uncontrollable, or for almost all controllable system numerical realizations, the there is no such a perturbation (in this paper, we only focus on the single-input case; the multi-input case will soon appear in our future research).  This paper presents a necessary and sufficient condition for a single-input system to be PTSC, whose verification has polynomial time complexity. Applications of our results on the structured controllability radius problems are also discussed.

1) We propose a novel notion PTSC to study controllability preservation for a single-input structured system under structured numerical perturbations. In PTSC, the perturbation structure can be {\emph{arbitrary}} relative to the structure of the original system. This notion provides a new view in studying the robustness of structural controllability other than structural perturbations.

2) We have shown PTSC can characterize the generic property that,  depending on the structure of the original system and the perturbations, for almost all of the controllable system realizations, either they can preserve controllability under arbitrary complex-valued perturbations with the given structure, or there is a perturbation with the given structure that can make the corresponding system fail to be controllable. This is beneficial in studying the SCRPs from a generic view.

3) We give a necessary and sufficient condition for a single-input system to be PTSC, whose verification has polynomial time complexity. The derivation is
based on the one-edge preservation principle and a series of nontrivial results on the roots of determinants of generic matrix pencils.

The rest is organized as follows. Section II introduces the PTSC notion and proves the involved genericity for single-input systems. Section III presents some preliminaries required for our further derivations. Section IV gives a necessary and sufficient condition for a single-input system to be PTSC. Section V discusses the application of PTSC on the SCRP. The last section concludes this paper.

Natations: Given an integer $p\ge 0$, define ${\cal J}_p\doteq \{1,...,p\}$. For a $p\times q$ matrix $M$, $M[{\cal I}_1,{\cal I}_2]$ denotes the submatrix of $M$ whose rows are indexed by ${\cal I}_1$ and columns by ${\cal I}_2$, ${\cal I}_1\subseteq {\cal J}_p$, ${\cal I}_2\subseteq {\cal J}_q$. For a vector $b$, $b_i$ denotes the $i$th entry of $b$.

%\section{Preliminaries}

\section{The Notion of PTSC}\label{ProblemFormulation}
\subsection{Structured Matrix}
Before presenting the notion of PTSC, we first introduce the so-called structured matrix. A structured matrix is a matrix whose entries are either fixed zero (denoted by $0$) or indeterminate parameters (denoted by $*$). For description simplicity, we may simply say the entry represented by $*$ is a nonzero entry. Let $\{0,*\}^{p\times q}$ be the set of all $p\times q$ dimensional structured matrices. For $\bar M \in \{0,*\}^{p\times q}$, the following two sets of matrices are defined:
$${\bf S}_{\bar M}=\left\{ M\in {\mathbb C}^{p\times q}: M_{ij}=0\ {\text{if}}\ {\bar M}_{ij}=0 \right\},$$
$$\bar {\bf S}_{\bar M}\!=\!\left\{ M\in {\mathbb C}^{p\times q}: M_{ij}=0 \ {\text{if}}\ {\bar M}_{ij}=0, M_{ij}\ne 0 \ {\text{if}}\ {\bar M}_{ij}=* \right\}.$$
Any $M\in {\bf S}_{\bar M}$ is called a realization of $\bar M$. For two structured matrices $\bar M, \bar N \in \{0,*\}^{p\times q}$, $\vee$ is the entry-wise OR operation, i.e., $(\bar M\vee \bar N)_{ij}= *$ if $\bar M_{ij}=*$ or $\bar N_{ij}=*$; otherwise $(\bar M\vee \bar N)_{ij}=0$.

A structured matrix could also be seen as a matrix whose entries are parameterized by the free parameters of its $*$ entries, and therefore is sometimes called a {\emph{generic matrix}} \cite{Murota_Book}.  For a generic matrix $M$ and a constant matrix $N$ with the same dimension, $M+\lambda N$ defines a generic matrix pencil, which can be seen as a {\emph{matrix-valued polynomial of free parameters in $M$ and the variable $\lambda$}}.

\subsection{ Notion of PTSC}
Consider the linear time invariant (LTI) system
\begin{equation}\label{plant}\dot x(t)=Ax(t)+bu(t),\end{equation}
where $A\in {\mathbb C}^{n\times n}$, $b\in {\mathbb C}^{n\times 1}$. It is known that $(A,b)$ is controllable, if and only if the controllability matrix ${\cal C}(A,  b)$ defined as follows has full row rank.
$${\cal C}(A,b)=[b,Ab,\cdots, A^{n-1}b].$$

Let ${\bar F}\in \{0,*\}^{n\times {(n+1)}}$ be a structured matrix that specifies the structure
of the perturbation (matrix) $[\Delta A, \Delta b]$, that is, ${\bar F}_{ij}=0$ implies $[\Delta A, \Delta b]_{ij}=0$. In other words, $[\Delta A, \Delta b]\in {\bf S}_{\bar F}$. It is emphasized that throughout this paper, the perturbations are allowed to be {\emph {complex-valued}. We will also call the system $(A+\Delta A, b+\Delta b)$ the perturbed system.

\begin{definition}[PTC]
System $(A,b)$ in (\ref{plant}) is said to be {\emph{perturbation-tolerantly controllable}} (PTC) with respect to ${\bar F}$, if for all $(\Delta A,\Delta b)\in {\bf S}_{\bar F}$, $(A+\Delta A, b+\Delta b)$ is controllable. If $(A,b)$ is controllable but not PTC w.r.t. ${\bar F}$ (i.e., there exists a $(\Delta A, \Delta b)\in {\bf S}_{\bar F}$ making $(A+\Delta A, b+\Delta b)$ uncontrollable), $(A,b)$ is said to be {\emph{perturbation-sensitively controllable}} (PSC) w.r.t. ${\bar F}$.
\end{definition}

Let $\bar A \in \{0,*\}^{n\times n}, \bar b\in \{0,*\}^{n\times 1}$ be the structured matrices specifying the sparsity pattern of $A,b$, respectively. That is, $A\in {\bf S}_{
\bar A}$ and $b\in  {\bf S}_{\bar b}$.

\begin{definition}[Structural controllability] $(\bar A, \bar b)$ is said to be structurally controllable, if there exists a realization $(A,b)\in {\bf S}_{[\bar A, \bar b]}$ so that $(A,b)$ is controllable.
\end{definition}

A property is called generic for a set of systems if, depending on the (common) structure of parameterized systems in this set, either this property holds for almost all of the system parameters in the corresponding parameter space, or this property does not hold for almost all of the system parameters \cite{generic}. It is well-known that controllability is a generic property in the sense that, if $(\bar A, \bar b)$ is structurally controllable, then all realizations of $(\bar A, \bar b)$ are controllable except for a set with zero Lebesgue measure in the corresponding parameter space. For a structurally controllable pair $(\bar A, \bar b)$, let ${\bf CS}(\bar A,\bar b)$ denote the set of all {\emph{controllable}} complex-valued realizations of $(\bar A, \bar b)$. The following proposition reveals that PTC (or PSC) is a generic property in ${\bf CS}(\bar A, \bar b)$.

\begin{proposition}\label{generic_PSC} With $(\bar A, \bar b)$ and ${\bar F}$ defined above, suppose that $(\bar A, \bar b)$ is structurally controllable. Then, either for all $(A,b)\in {\bf CS}(\bar A,\bar b)$, $(A,b)$ is PTC w.r.t. ${\bar F}$, or for almost all $(A,b)\in {\bf CS}(\bar A,\bar b)$ except for a set with zero Lebesgue measure in the corresponding parameter space, $(A,b)$ is PSC w.r.t. ${\bar F}$.
\end{proposition}

\begin{proof} Let $p_1,...,p_r$ be  variables that the $r$ nonzero entries of $[A,b]$ take, and $\bar p_1,...,\bar p_l$ be variables that the $l$  perturbed entries of $[\Delta A, \Delta b]$ take. Denote by $p\doteq (p_1,...,p_r)$ and $\bar p\doteq (\bar p_1,...,\bar p_l)$. It turns out $\det {\cal C}(A+\Delta A, b+\Delta b)$ can be expressed as
\begin{equation} \label{determinant_of_Ab}
 \det {\cal C}(A+\Delta A, b+\Delta b)=f(p)g(\bar p)h(p,\bar p),
\end{equation}
where $f(p)$ (resp. $g(\bar p)$) denotes the polynomial of $p$ (resp. $\bar p$) with real coefficients, and $h(p,\bar p)$ denotes the polynomial of $p$ and $\bar p$ where at least one $p_i$ ($i\in \{1,...,r\}$) and one $\bar p_j$ ($j\in \{1,...,l\}$) both have a term with degree no less than one.

It can be seen that, if neither $g(\bar p)$ nor $h(p,\bar p)$ exists in (\ref{determinant_of_Ab}), then for all $(A,b)\in {\bf CS}(\bar A,\bar b)$, $(A,b)$ is PTC w.r.t. ${\bar F}$, as in this case, (\ref{determinant_of_Ab}) is independent of $\bar p$. Otherwise, suppose that there exists a $\bar p_j$, $j\in\{1,...,l\}$, that has a term $\bar p_j^k$ with degree $k\ge 1$ in $g(\bar p)h(p,\bar p)$. When $p$ and $\bar p \backslash \{\bar p_j\}$ take such values that the coefficient of $\bar p_j^k$ is nonzero (Condition 1)), and meanwhile $(A,b)$ is controllable (Condition 2)), then according to the fundamental theorem of algebra (c.f. \cite{dummit2004abstract}), there exists a complex value of $\bar p_j$ making $g(\bar p)h(p,\bar p)=0$, thus making $\det {\cal C}(A+\Delta A, b+\Delta b)=0$, leading to the uncontrollability of $(A+\Delta A, b+\Delta b)$. Note that the set of values for $p$ validating Condition 1) forms a proper algebraic variety in ${\mathbb C}^{r}$, which thereby has zero Lebesgue measure in the set of values satisfying Condition 2). This proves the proposed statements.
\end{proof}

Proposition \ref{generic_PSC} indicates that it is the structure of $(A,b)$ and ${\bar F}$ that domains the property of being PTC or PSC. This motivates the following definition for structural controllability under structured numerical perturbations: % for a system w.r.t. a prescribed perturbation structure.

\begin{definition}[PTSC]\label{DefinitionForPTSC} Given $\bar A\in \{0,*\}^{n\times n}$, $\bar b\in \{0,*\}^{n\times 1}$, and a perturbation structure ${\bar F}\in \{0,*\}^{n\times (n+1)}$, $(\bar A, \bar b)$ is said to be PTSC with respect to ${\bar F}$, if for all $(A,b)\in {\bf CS}(\bar A,\bar b)$, $(A,b)$ is PTC w.r.t. ${\bar F}$.
\end{definition}

If $(\bar A,\bar b)$ is structurally controllable but not PTSC w.r.t. ${\bar F}$, we will alternatively say that $(\bar A,\bar b)$ is perturbation-sensitively structurally controllable (PSSC) w.r.t. ${\bar F}$. From Proposition \ref{generic_PSC}, PSSC of $(\bar A, \bar b)$ indicates that almost all controllable realizations of $(\bar A, \bar b)$ are PSC w.r.t. the perturbation structure ${\bar F}$.

\begin{example}\label{exp1}Consider a system $(A,b)$ with
$$
A=\left[
    \begin{array}{cccc}
      0 & 0 & 0 & 0 \\
      a & 0 & 0 & 0 \\
      0 & c & 0 & 0 \\
      f & d & 0 & e \\
    \end{array}
  \right],b=\left[
              \begin{array}{c}
                h \\
                0 \\
                0 \\
                0 \\
              \end{array}
            \right].
$$Two perturbations $[\Delta A_i, \Delta b_i]$ ($i=1,2$) are given as {\small
$$\left[
    \begin{array}{ccccc}
      0 & 0 & g & l & 0 \\
      0 & 0 & 0 & 0 & 0 \\
      0 & 0 & 0 & 0 & 0 \\
      0 & 0 & 0 & 0 & 0 \\
    \end{array}
  \right],\left[
    \begin{array}{ccccc}
      0 & 0 & 0 & 0 & 0 \\
      0 & 0 & 0 & 0 & 0 \\
      0 & 0 & s & 0 & 0 \\
      0 & 0 & 0 & 0 & r \\
    \end{array}
  \right].
$$}It can be obtained $\det {\cal C}(A+\Delta A_1, b+\Delta b_1)=a^2ch^4(fe^2 + ade)$, which is independent of the parameters in $[\Delta A_1, \Delta b_1]$. This indicates whatever values $[\Delta A_1, \Delta b_1]$ may take, the perturbed system is always controllable (as long as $(A,b)$ is controllable). On the other hand,
$\det {\cal C}(A+\Delta A_2, b+\Delta b_2)=a^2ch^3(e^3r - e^2rs + e^2fh - efhs + adeh - adhs)$. It can be seen that, in case $e^3 - se^2\ne 0$ (or $re^2+fhe+adh\ne 0$) and $(A,b)$ is controllable, there is a value for $r$ (resp. $s$) making $\det {\cal C}(A+\Delta A_2, b+\Delta b_2)=0$, leading to uncontrollability.

\end{example}

\begin{remark}  We remark that
for multi-input systems,  PTC is still a generic property, but one should modify `all' to `almost all' in the last sentence of Definition \ref{DefinitionForPTSC} in defining the associated PTSC. Please see \cite{full-version-tac} for more details.
\end{remark}

\subsection{Relations with SSC}
We may also revisit PTSC from the standpoint of the {\emph{perturbed structured system}} $[\bar A, \bar b]\vee \bar F$. In this system, entries of system matrices can be divided into three categories, namely, the {\emph{fixed zero entries}}, the {\emph{unknown
generic entries}} which take fixed but unknown values (they can be seen as randomly generated values), and the {\emph{perturbed entries}} which can take arbitrarily complex values. PTSC of the perturbed structured system requires that for almost all values of the unknown generic entries making the original system controllable, the corresponding perturbed systems are controllable for arbitrary values of the perturbed entries.

 Recall that SSC is defined as follows.
\begin{definition}[SSC] $(\bar A, \bar b)$ is said to be SSC, if every $(A,b)\in \bar {\bf S}_{[\bar A, \bar b]}$ is controllable.
\end{definition}

As mentioned earlier, SSC could be seen as the ability of a system to preserve controllability under perturbations that have the same structure as the system itself, with the constraint that the perturbed entries of the resulting system cannot be zero. It is thus clear that the essential difference between PTSC and SSC lies in two aspects: First, the perturbed entries can take arbitrary values including zero in PTSC, while they must take nonzero values in SSC. Second, in SSC, all nonzero entries can be perturbed, while in PTSC, an arbitrary subset of entries (prescribed by the perturbation structure) can be perturbed, and the remaining entries remain unchanged. Because of them, neither criteria for SSC can be converted to those for PTSC, nor the converse.

\section{Preliminaries and Terminologies}
In this section, we introduce some preliminaries as well as terminologies in graph theory and structural controllability.
\subsection{Graph Theory}
If not specified, all graphs in this paper refer to directed graphs. A graph is denoted by ${\cal G}=({\cal V},{\cal E})$, where ${\cal V}$ is the vertex set, and ${\cal E}\subseteq {\cal V}\times {\cal V}$ is the edge set. For a graph ${\cal G}=({\cal V},{\cal E})$ with $N$ vertices, a path from vertex $v_i$ to vertex $v_j$ is a sequence of edges $(v_i,v_{i+1})$, $(v_{i+1},v_{i+2})$, $\cdots$, $(v_{j-1},v_{j})$ where each edge belongs to $\cal E$.  For a set ${\cal E}_s\subseteq {\cal E}$, ${\cal G}-{{\cal E}_s}$ denotes the graph obtained from $\cal G$ after deleting the edges in ${\cal E}_s$; similarly, for ${\cal V}_s\subseteq {\cal V}$, ${\cal G}-{\cal V}_s$ denotes the graph after deleting vertices in ${\cal V}_s$ and all edges incident to vertices in ${\cal V}_s$. For two graphs ${\cal G}_i=({\cal V}, {\cal E}_i)$, $i=1,2$, ${\cal G}_1\cup {\cal G}_2$ denotes the graph $({\cal V}, {\cal E}_1\cup {\cal E}_2)$.

A graph ${\cal G}=({\cal V},{\cal E})$ is said to be bipartite if its vertex set can be divided into two disjoint parts ${\cal V}_1$ and ${\cal V}_2$ such that no edge has two ends within ${\cal V}_1$ or ${\cal V}_2$. The bipartite graph $\cal G$ is also denoted by $({\cal V}_1, {\cal V}_2, {\cal E})$.  A matching of a bipartite graph is a subset of its edges among which any two do not share a common vertex. The maximum matching is the matching with the largest number of edges among all possible matchings. The number of edges contained in a maximum matching of a bipartite graph $\cal G$ is denoted by ${\rm mt}({\cal G})$. For a weighted bipartite graph $\cal G$, where each edge is assigned a non-negative weight, the weight of a matching is the sum of all edges contained in this matching. The minimum weight maximum matching (resp. maximum weight maximum matching) is the minimal weight (resp. maximal weight) over all maximum matchings of ${\cal G}$.

The generic rank of a structured matrix $\bar M$, given by  ${\rm grank}(\bar M)$, is the maximum rank $\bar M$ can achieve as the function of its free parameters.
The bipartite graph associated with a structured matrix $\bar M$ is given by ${\mathcal B}(\bar M)=(R,C,{\mathcal E})$, where the left (resp. right) vertex set $R$ ($C$) corresponds to the row index (column index) set of $\bar M$, and the edge set corresponds to the set of nonzero entries of $\bar M$, i.e., ${\mathcal E}=\{(i,j): i\in R, j\in C, \bar M_{ij}\ne 0\}$. It is known that, ${\rm grank}(\bar M)$ equals the cardinality of the maximum matching of ${\mathcal B}(\bar M)$.

{\subsection{DM-Decomposition}}
Dulmage-Mendelsohn decomposition (DM-decomposition for short) is a unique decomposition of a bipartite graph w.r.t. maximum matchings. Let
${\cal G}=({\cal V}^{+}, {\cal V}^{-}, {\cal E})$ be a bipartite graph. For ${\cal M}\subseteq {\cal E}$, we denote by $V^{+}({\cal M})$ (resp. $V^{-}({\cal M})$) the set of vertices in ${\cal V}^+$
(resp. ${\cal V}^{-}$) incident to edges in ${\cal M}$.  An edge of $\cal G$ is said to be admissible,
if it is contained in some maximum matching of $\cal G$.

\begin{definition}[DM-decomposition,\cite{Murota_Book}] \label{DM-def} The DM-decomposition of a bipartite graph
 ${\cal G}=({\cal V}^{+}, {\cal V}^{-}, {\cal E})$ is to decompose $\cal G$ into subgraphs ${\cal G}_k=({\cal V}^+_k,{\cal V}^-_k,{\cal E}_k)$ ($k=0,1,...,d,\infty$)
 (called DM-components of $\cal G$) satisfying:

 1) ${\cal V}^{\star}=\bigcup \nolimits_{k=0}^{\infty} {\cal V}^{\star}_k$, ${\cal V}^{\star}_i\bigcap {\cal V}^{\star}_j=\emptyset$ for $i\ne j$,  with $\star= +$ and $-$; ${\cal E}_k=\{(v^+,v^-)\in {\cal E}: v^+\in {\cal V}^+_k, v^-\in {\cal V}^-_k\}$;

 2) For $1\le k \le d$ (consistent components): ${\rm mt}({\cal G}_k)=|{\cal V}_k^+|=|{\cal V}_k^-|$, and each $e\in {\cal E}_k$ is admissible in ${\cal G}_k$;
 for $k=0$ (horizontal tail): ${\rm mt}({\cal G}_0)=|{\cal V}_0^-|$, $|{\cal V}_0^+|<|{\cal V}_0^-|$ if ${\cal V}_0^+\ne \emptyset$, and each $e\in {\cal E}_0$ is
 admissible in ${\cal G}_0$; for $k=\infty$ (vertical tail): ${\rm mt}({\cal G}_{\infty})=|{\cal V}^+_\infty|$, $|{\cal V}_\infty^+|>|{\cal V}_\infty^-|$ if ${\cal V}_\infty^-\ne \emptyset$, and each $e\in {\cal E}_\infty$ is
 admissible in ${\cal G}_0$;

 3) %The partial order $\preccurlyeq$ among ${\cal G}_k$ is represented by the existence of edges: ${\cal E}_{kl}=\emptyset$ unless ${\cal G}_k \preceq {\cal G}_l$,
% $1\le k,l\le b$, and ${\cal E}_{kl}\ne \emptyset$ if ${\cal G}_k \prec {\cal G}_l$,
% $1\le k,l\le b$, where ${\cal E}_{kl}=\{e\in {\cal E}: V^+(e)\in {\cal V}_k, V^-(e)\in {\cal V}_l\}$,
${\cal E}_{kl}=\emptyset$ unless $1\le k\le l\le d$, and ${\cal E}_{kl}\ne \emptyset$ only if
 $1\le k\le l\le d$, where ${\cal E}_{kl}=\{e\in {\cal E}: V^+(\{e\})\in {\cal V}_k, V^-(\{e\})\in {\cal V}_l\}$;

 4) ${\cal G}$ cannot be decomposed into more components satisfying conditions 1)-3).
\end{definition}

For an $m\times l$ matrix $M$, the DM-decomposition of ${\cal B}(M)$ into diagraphs ${\cal G}_k=({\cal V}^+_k,{\cal V}^-_k,{\cal E}_k)$ ($k=0,1,...,d,\infty$) corresponds to that there exist two permutation matrices $P\in {\mathbb R}^{m\times m}$ and $Q\in {\mathbb R}^{l\times l}$ satisfying
\begin{equation}\label{DMmatrix}
PMQ=\left[
      \begin{array}{ccccc}
        M_0 & M_{01} & \cdots & M_{0d} & M_{0\infty} \\
        0 & M_1 & \cdots & M_{1d} & M_{1\infty} \\
        0 & 0 & \ddots & \cdots & \cdots \\
        0 & 0 & \cdots & M_d & M_{d\infty} \\
        0 & 0 & \cdots & 0 & M_\infty \\
      \end{array}
    \right],
\end{equation}where the submatrix $M_k=M[{\cal V}^+_k,{\cal V}^-_k]$ corresponds to ${\cal G}_k$ ($k=0,1,...,d,\infty$). Matrix (\ref{DMmatrix}) is also called the DM-decomposition of $M$.

A bipartite graph is said to be DM-irreducible if it cannot be decomposed into more than one nonempty component in the DM-decomposition.
 DM-decomposition is closely related to the irreducibility of the determinant of a generic matrix. Recall that a multivariable polynomial $f$ is irreducible if it cannot be factored as $f=f_1f_2$ with $f_1,f_2$ being polynomials with smaller degrees than $f$.

\begin{lemma} \citep[Theorems 2.2.24, 2.2.28]{Murota_Book}\label{reduciblility} For a bipartite graph ${\cal G}(M)=({\cal V}^+ ,{\cal V}^- ,{\cal E}(M))$ associated with a generic square matrix $M$, the following conditions are equivalent:

1) ${\cal G}(M)$ is DM-irreducible;

2) ${\rm mt}({\cal G}(M)-\{v_1,v_2\})= {\rm mt}({\cal G}(M))-1$ for any $v_1\in {\cal V}^+$ and $v_2\in {\cal V}^-$;

3) $\det M$ is irreducible.
\end{lemma}

\subsection{Structural Controllability}
Given $(\bar A,\bar b)$, let $\mathcal{X}$, $\mathcal{U}$ denote the sets of state vertices and input vertices respectively, i.e., $\mathcal{X}=\{x_1,...,x_n\}$, $\mathcal{U}=\{{x_{n+1}}\}$. Denote the edges by $\mathcal{E}_{\mathcal{X},\mathcal{X}}(\bar A)=\{(x_i,x_j): \bar A_{ji}\ne 0\}$, $\mathcal{E}_{\mathcal{U},\mathcal{X}}(\bar b)=\{({x_{n+1}},x_i): \bar b_{i}\ne 0\}$.  Let $\mathcal{G}(\bar A,\bar b)= (\mathcal{X}\cup \mathcal{U}, \mathcal{E}_{\mathcal{X},\mathcal{X}}(\bar A)\cup \mathcal{E}_{\mathcal{U},\mathcal{X}}(\bar b))$ be the system graph associated with $(\bar A, \bar b)$. A state vertex $x\in\mathcal{X}$ is said to be input-reachable, if there exists a path from an input vertex ${x_{n+1}}$ to $x$ in $\mathcal{G}(\bar A,\bar b)$. Similarly, ${\cal G}({\bar F})=({\cal X}\cup {\cal U}, {\cal E}_{\bar F})$ denotes the graph associated with the perturbation structure ${\bar F}$, where ${\cal E}_{\bar F}=\{(x_i,x_j): {\bar F}_{ji}\ne 0, 1\le j\le n, 1\le i \le n+1\}$.
\begin{lemma}\cite{Murota_Book} \label{theo-strucon} System $(\bar A, \bar b)$ in (\ref{plant}) is structurally controllable, if and only if
1) every state vertex $x\in {\cal X}$ is input-reachable; 2) ${\rm grank}([\bar A,\bar b])=n$.
\end{lemma}
\vspace{-0.1cm}
\section{Necessary and Sufficient Condition}
In this section, we present a necessary and sufficient condition for the PTSC in the single-input case.
\vspace{-0.2cm}
\subsection{One-edge Preservation Principle}
At first, a one-edge preservation principle is given as follows, which is fundamental to our subsequent derivations.

\begin{proposition}[One-edge preservation principle] \label{OneEdgeEquivalence} Suppose $(\bar A, \bar b)$ is structurally controllable. $(\bar A,\bar b)$ is PSSC w.r.t. $\bar F$, if and only if there is one edge $e\in {\cal E}_{\bar F}$, such that $[\bar A, \bar b]\vee \bar F^{\{e\}}$ is PSSC w.r.t. ${\bar F}_{\{e\}}$, where $\bar F^{\{e\}}$ denotes the structured matrix associated with the graph ${\cal G}(\bar F)-\{e\}$, and ${\bar F}_{\{e\}}$ the structured matrix obtained from ${\bar F}$ by preserving only the entry corresponding to $e$. % $[\bar A, \bar b]\vee {\bar F}^{\{e\}}$
\end{proposition}

\begin{proof} Let $p\doteq (p_1,...,p_r)$ and $\bar p=(\bar p_1,...,\bar p_l)$ be defined similarly in the the proof of Proposition \ref{generic_PSC}. From the analysis in that proof, $(\bar A, \bar b)$ is PSSC w.r.t. ${\bar F}$, if and only if there exists a $\bar p_j$, $j\in \{1,...,l\}$, that has a degree no less than one in $\det {\cal C}(A+\Delta A, b+\Delta b)$ (expressed in (\ref{determinant_of_Ab})),  where $(A,b)$ and $(\Delta A, \Delta b)$ are realizations of $(\bar A, \bar b)$ and $(\Delta A, \Delta b)$ respectively, with the corresponding parameters being $p$ and $\bar p$. Let $e$ be the edge corresponding to $\bar p_j$. Suppose that the coefficient of $\bar p_j^k$ is nonzero for some degree $k\ge 1$. Since the coefficient of $\bar p_j^k$ is a polynomial of $p$ and ${\bar p}\backslash \{\bar p_j\}$ in $\det {\cal C}(A+\Delta A, b+\Delta b)$, it always equals that in $\det {\cal C}(A'+\Delta A', b'+\Delta b')$, where $(A',b')$ has the system graph ${\cal G}(\bar A, \bar b)\cup {\cal G}({\bar F})- \{e\}$, and $[\Delta A',\Delta b']$ corresponds to the perturbation ${\bar F}_{\{e\}}$, noting that $[A+\Delta A, b+\Delta b]=[A'+\Delta A', b'+\Delta b']$ in the symbolic operation sense. Upon observing this, the proposed statement follows immediately.
\end{proof}

It is remarkable that the one-edge preservation principle does {\emph{not}} mean the perturbation of only one entry is enough to destroy controllability; Instead, it means we can regard $|{\cal E}_{\bar F}|-1$ entries of $\bar F$ as unknown generic entries ({\emph{in other words, their values can be chosen randomly; but not fixed zero}}) and find suitable values for the last entry. This principle indicates that verifying the PTSC w.r.t. an arbitrary perturbation structure can be reduced to an equivalent problem with a single-edge perturbation structure. Having observed this, in the following, we will give the conditions for the absence of zero uncontrollable modes and nonzero uncontrollable modes, respectively, in the single-edge perturbation scenario. Recall that an uncontrollable mode for $(A,B)$ is a $\lambda \in {\mathbb C}$ making ${\rm rank}([A-\lambda I, B])<n$. Then, based on Proposition \ref{OneEdgeEquivalence}, conditions for PTSC with a general perturbation structure will be obtained.

\subsection{Condition for Zero Mode}
Let $\bar H\doteq [\bar A, \bar b]$. For $j\in \{1,...,n+1\}$, let $r_j={\rm grank}(\bar H[{\cal J}_{n},{\cal J}_{n+1}\backslash \{j\}])$. Define  sets ${\cal I}_j$ and ${\cal I}^*_j$ as $$\begin{array}{c}  \begin{aligned}{\cal I}_j&=\left\{{\cal I}\subseteq {\cal J}_n: {\rm grank} (\bar H[{\cal I},{\cal J}_{n+1}\backslash \{j\}])=r_j, |{\cal I}|=r_j\right\},\\ {\cal I}^{*}_j&=\left\{{\cal J}_n\backslash w: w\in {\cal I}_j\right\}.  \end{aligned}\end{array} $$Based on these definitions, the following proposition gives a necessary and sufficient condition for the absence of zero uncontrollable modes in the single-edge perturbation scenario.

\begin{proposition}\label{ZeroCondition} Suppose that $(\bar A, \bar b)$ is structurally controllable, and there is only one nonzero entry in ${\bar F}$ with its position being $(i,j)$. Then, for almost all $(A,b)\in {\bf CS}(\bar A,\bar b)$, there is no $(\Delta A, \Delta b)\in {\bf S}_{{\bar F}}$ such that a nonzero $n$-vector $q$ exists making $q^{\intercal}[A+\Delta A, b+\Delta b]=0$, if and only if $i\notin {\cal I}^{*}_j$.
\end{proposition}

To prove Proposition \ref{ZeroCondition}, we need the following lemma.

\begin{lemma}\label{basiclemma} Given a matrix $H\in {\mathbb C}^{p\times q}$ of rank $p-1$, let $x\in {\mathbb C}^{p}$ be a nonzero vector in the left null space of $H$. Then, for any $i\in {\cal J}_p$,  $x_i\ne 0$, if and only if $H[{{\cal J}_p\backslash \{i\},{\cal J}_q}]$ is  of full row rank.
%suppose that $T\in {\mathbb C}^{r\times p}$ consists of the maximum number of independent {\emph{row}} vectors which span the left null space of $H$. Then, for any ${\cal I}\subseteq {\cal J}_p$, $T[{\cal J}_r, {\cal I}]$ is full of row rank, if and only if $H[{{\cal J}_p\backslash {\cal I},{\cal J}_q}]$ is full of row rank.
\end{lemma}

%Given a matrix $H\in {\mathbb C}^{p\times q}$ of rank $p-1$, let $x\in {\mathbb C}^{}$
%suppose that $T\in {\mathbb C}^{r\times p}$ consists of the maximum number of independent {\emph{row}} vectors which span the left null space of $H$. Then, for any ${\cal I}\subseteq {\cal J}_p$, $T[{\cal J}_r, {\cal I}]$ is full of row rank, if and only if $H[{{\cal J}_p\backslash {\cal I},{\cal J}_q}]$ is full of row rank.

\begin{proof} The proof is quite standard, thus omitted here.
%Without losing any generality, consider $i=1$. Let $h_1=H[\{1\},{\cal J}_q]$, $h_{2:p}=H[{\cal J}_p\backslash \{1\},{\cal J}_q]$, and accordingly, let $x_{2:p}=[x_2,...,x_p]$. By definition,
%\begin{equation}\label{nulleq} x^\intercal H=[x_1,x_{2:p}][h^\intercal_1,h^\intercal_{2:p}]^\intercal=x_1h_1+x_{2:p}h_{2:p}=0.\end{equation}
% {\bf If direction}:  Suppose $h_{2:p}$ is of full  row rank but $x_1=0$. Then, from (\ref{nulleq}), $x_{2:p}h_{2:p}=0$. As $x_2\ne 0$ (otherwise $x=0$), $h_{2:p}$ is of row rank deficient, causing a contradiction. {\bf Only if direction}: Suppose $x_1\ne 0$. If $h_{2:p}$ is of row rank deficient, then there is $y\in {\mathbb C}^{p-1}$ with $y\ne0$ making $y^{\intercal} h_{2:p}=0$. Hence, $[0,y^{\intercal}]^\intercal$ is also in the left null space of $H$. As $x_1\ne 0$, this contradicts the fact $H$ has rank $p-1$.
\end{proof}

%Note that ${\cal J}_p\backslash {\cal I}$ is a complement of ${\cal I}$ in ${\cal J}_p$. Lemma \ref{basiclemma} establishes a correspondence in rank between a subset of columns of $T$ and its complement of rows of $H$.

{\bf {\emph {Proof of Proposition \ref{ZeroCondition}:}}} Sufficiency:  Since $(\bar A, \bar b)$ is structurally controllable, ${\rm grank}(\bar H)= n$, which means that $r_j=n$ or $n-1$. If $r_j=n$, then ${\cal I}^*_j=\emptyset$, which immediately indicates that no $q(\ne 0)$ exists making $q^{\intercal}[A+\Delta A, b+\Delta b]=0$ for almost all $(A,b)\in {\bf CS}(\bar A,\bar b)$. Now suppose $r_j=n-1$. A vector $q (\ne 0)$ making $q^{\intercal}[A+\Delta A, b+\Delta b]$ must lie in the left null space of $H[{\cal J}_{n}, J_{n+1}\backslash \{j\}]$, for almost all $H\in {\bf S}_{\bar H}$. As $r_j=n-1$, for almost all $(A,b)\in {\bf CS}(\bar A,\bar b)$, ${\cal I}^*_j$ consists of all the nonzero positions of $q$ according to Lemma \ref{basiclemma}.  As a result, if $i\notin {\cal I}^*_j$,
$$q^{\intercal}[A+\Delta A, b+\Delta b][{\cal J}_{n},\{j\}]=\sum \nolimits_{k\in {\cal I}^*_j}q_k[A,b]_{kj}\ne 0,$$
where the inequality is due to the fact that $[A,b]$ has full row rank.

Necessity: Assume that $i\in {\cal I}^*_j$. As ${\cal I}^*_j \ne \emptyset$ and $(\bar A, \bar b)$ is structurally controllable, $[A,b][{\cal J}_n,{\cal J}_{n+1}\backslash \{j\}]$ has rank $n-1$ for all $(A,b)\in {\bf CS}(\bar A,\bar b)$. Let $q$ be a nonzero vector in the left null space of  $[A,b][{\cal J}_n,{\cal J}_{n+1}\backslash \{j\}]$.
According to Lemma \ref{basiclemma}, as $i\in {\cal I}^*_j$, we have $q_{i}\ne 0$ for almost all $(A,b)\in {\bf CS}(\bar A,\bar b)$. By setting $[\Delta A, \Delta b]_{ij}=-1/q_i\sum \nolimits_{k\in {\cal I}^*_j\backslash \{i\}}q_k[A,b]_{kj}$, we get $$q^{\intercal}[A+\Delta A, b+\Delta b][{\cal J}_n,\{j\}]=\sum \nolimits_{k\in {\cal I}^*_j}q_k[A,b]_{kj}=0,$$
which makes $q^{\intercal}[A+\Delta A, b+\Delta b]=0$.   $\hfill \Box$

\begin{remark}\label{equal-pro1} From the proof of Proposition \ref{ZeroCondition}, provided $(\bar A, \bar b)$ is structurally controllable, $i\notin {\cal I}_j^*$ is equivalent to that, ${\rm grank}(\bar H[{\cal J}_n, {\cal J}_{n+1}\backslash \{j\}])=n$ (corresponding to $r_j=n$) or ${\rm grank}(\bar H[{\cal J}_n\backslash \{i\}, {\cal J}_{n+1}\backslash \{j\}])=n-2$ (corresponding to $r_j=n-1$ but ${\rm grank}(\bar H[{\cal J}_n\backslash \{i\}, {\cal J}_{n+1}\backslash \{j\}])<n-1$). Moreover, since adding a column to a matrix can increase its rank by at most one,  the latter two conditions are mutually exclusive.
\end{remark}
\vspace{-0.3cm}
\subsection{Condition for Nonzero Mode} \label{sec-nonzero}
In the following, we present a necessary and sufficient condition for the absence of nonzero uncontrollable modes using the DM-decomposition.

For $j\in {\cal J}_{n+1}$, let $j_{\rm c}\doteq {\cal J}_{n+1}\backslash \{j\}$. Moreover, define a generic matrix pencil as $H_\lambda\doteq [\bar A-\lambda I,\bar b]$, $H_\lambda^j\doteq H_{\lambda}[{\cal J}_n,\{j\}]$, and $H_\lambda^{{j_{\rm c}}}\doteq H_{\lambda}[{\cal J}_n,j_c]$. Here, the subscript $\lambda$ indicates a matrix-valued function of $\lambda$.  Let ${\cal B}(H_\lambda)=({\cal V}^+,{\cal V}^-,{\cal E})$ be the bipartite graph associated with $H_\lambda$, where ${\cal V}^+=\{x_1,...,x_n\}$, ${\cal V}^-
=\{v_1,...,v_{n+1}\}$, and ${\cal E}=\{(x_i,v_k): {\cal E}_I\cup {\cal E}_{[\bar A,\bar b]}\}$ with ${\cal E}_I=\{(x_i,v_i):i=1,...,n\}$, ${\cal E}_{[\bar A,\bar b]}=\{(x_i,v_k):
[\bar A,\bar b]_{ik}\ne 0\}$. No parallel edges are included even if ${\cal E}_I\cap {\cal E}_{[\bar A,\bar b]}\ne \emptyset$. An edge is called a $\lambda$-edge if it belongs to ${\cal E}_I$,  and a self-loop if it belongs to ${\cal E}_I\cap {\cal E}_{[\bar A,\bar b]}$. Note by definition, a self-loop is also a $\lambda$-edge. Let ${\cal B}(H_\lambda^{{j_{\rm c}}})$ be the bipartite graph associated with $H_\lambda^{{j_{\rm c}}}$, that is, ${\cal B}(H_\lambda^{{j_{\rm c}}})= {\cal B}(H_{\lambda})-\{v_j\}$.
%({\cal V}^+,{\cal V}^-\backslash \{v_j\}, {\cal E}\backslash {\cal E}_{.j})$ with ${\cal E}_{.j}=\{(x_i,v_j):(x_i,v_j)\in {\cal E}\}$.

\begin{lemma} \label{fullrank}
Suppose $(\bar A,\bar b)$ is structurally controllable. Then ${\rm mt}({\cal B}(H_\lambda^{{j_{\rm c}}}))=n$ for all $j\in {\cal J}_{n+1}$.
\end{lemma}

\begin{proof} If $j=n+1$, it is obvious ${\rm mt}({\cal B}(H_\lambda^{{j_{\rm c}}}))=n$ as ${\cal E}_I$ is a matching with size $n$. Now consider $j\in\{1,...,n\}$. As  $(\bar A,\bar b)$ is structurally controllable, from Lemma \ref{theo-strucon}, there is a path from $x_{n+1}$ to $x_j$ in the system graph ${\cal G}(\bar A, \bar b)$. Denote such a path by $\{(x_{n+1},x_{j_1}),(x_{j_1},x_{j_2}),...,(x_{j_{r-1}},x_{j_r})\}$ with  $\{j_1,...,j_r\}\subseteq {\cal J}_n$ and $j_r=j$. Since each $(x_{j_k},x_{j_{k+1}})$ in ${\cal G}(\bar A, \bar b)$ corresponds to $(x_{j_{k+1}},v_{j_k})$ in ${\cal B}(H_\lambda)$, $\{(x_{j_1},v_{n+1}),(x_{j_2},v_{j_1}),...,(x_{j_r},v_{j_{r-1}})\}\cup \{(x_i,v_i): i\in {\cal J}_n\backslash \{j_1,...,j_r\}\}$ forms a matching with size $n$ in ${\cal B}(H_\lambda^{{j_{\rm c}}})$.
\end{proof}

Let ${\cal G}^{j_{\rm c}}_k=({\cal V}^+_k,{\cal V}_k^-,{\cal E}_k)$ ($k=0,1,...,d,\infty$) be the DM-components of ${\cal B}(H_\lambda^{{j_{\rm c}}})$. From Lemma \ref{fullrank}, we know that both the horizontal tail and the vertical one are empty. Accordingly, let $M^{j_{\rm c}}_\lambda$ be the DM-decomposition of $H_\lambda^{{j_{\rm c}}}$ with the corresponding permutation matrices $P$ and $Q$, i.e.,
\begin{equation}\label{DMdecomposition} PH_\lambda^{{j_{\rm c}}}Q=\left[
      \begin{array}{ccc}
       M^{j_{\rm c}}_1(\lambda) & \cdots & M^{j_{\rm c}}_{1d}(\lambda) \\
        0 & \ddots & \vdots \\
         0 & \cdots & M^{j_{\rm c}}_d(\lambda)  \\
      \end{array}
    \right]\doteq M^{j_{\rm c}}_\lambda.\end{equation}
Moreover, define $M_\lambda^j\doteq PH_\lambda^j$. Suppose that $x_i$ is the $\bar i$th vertex in ${\cal V}^+_{i^*}$ ($1\le \bar i\le |{\cal V}^+_{i^*}|$, $1\le i^* \le d$).

For $k\in \{1,...,d\}$, let $\gamma_{\min}({\cal G}_k^{j_{\rm c}})$ and $\gamma_{\max}({\cal G}_k^{j_{\rm c}})$ be respectively the minimum number of $\lambda$-edges and maximum number of $\lambda$-edges contained in a matching among all maximum matchings of ${\cal G}_k^{j_{\rm c}}$. Afterwards, define a boolean function $\gamma_{\rm nz}(\cdot)$ for ${\cal G}_k^{j_{\rm c}}$ as
\begin{equation}\label{fun-nz} \gamma_{\rm nz}({\cal G}_k^{j_{\rm c}})=
\begin{cases} 1 & {\begin{array}{c} \text{if}\ \gamma_{\max}({\cal G}_k^{j_{\rm c}})-\gamma_{\min}({\cal G}_k^{j_{\rm c}})>0 \\ \text{or} \ {\cal G}_k^{j_{\rm c}} \ {\text{contains a self-loop}} \end{array}} \\
0 &  \text{otherwise}.
\end{cases}\end{equation}

From Lemma \ref{matrix_pencil} in the appendix, $\gamma_{\rm nz}({\cal G}_k^{j_{\rm c}})=1$ means $\det M_k^j(\lambda)$ has at least one nonzero root for $\lambda$, while $\gamma_{\rm nz}({\cal G}_k^{j_{\rm c}})=0$ means the contrary. The following lemma shows that $\gamma_{\rm nz}({\cal G}_k^{j_{\rm c}})$ can be determined in polynomial time via the maximum/minimum weight maximum matching algorithms.

\begin{lemma} Assign the following weight function $W_k: {\cal E}_k \rightarrow \{0,1\}$ for ${\cal G}_k^{j_{\rm c}}$ as
$$W_k(e)=
\begin{cases} 1 & \text{if}\ e \ \text{is a} \ \lambda\text{-edge}\\
0 & \text{if}\ e\in {\cal E}_k\backslash {\cal E}_I.
\end{cases}$$
Then, it is true that $$\gamma_{\max}({\cal G}_k^{j_{\rm c}})= {\text{ the maximum weight maximum matching of}}\ {\cal G}_k^{j_{\rm c}},$$
$$\gamma_{\min}({\cal G}_k^{j_{\rm c}})= {\text{ the minimum weight maximum matching of}}\ {\cal G}_k^{j_{\rm c}}.$$
\end{lemma}

\begin{proof} Straightforward from the construction of $W_k$ and the definitions of the maximum (minimum) weight maximum matching of a bipartite graph.
\end{proof}

%  ${\cal G}^j_{ki^*}=\left({\cal V}^{+}_k\cup {\cal V}^{+}_{k+1}\cdots\cup {\cal V}^+_{i^*}\backslash \{x_i\}, {\cal V}^{-}_k\cup {\cal V}^{-}_{k+1}\cdots \cup {\cal V}^-_{i^*}, \bar {\cal E}_{ki^*}, W\right)$, where

Furthermore, define a set
\begin{equation}\label{important-set} {\Omega}_j=\{k\in {\mathbb N}: 1\le k\le i^*, {\gamma}_{\rm nz}({\cal G}^{j_{\rm c}}_k)= 1\}.\end{equation}
For each $k\in {\Omega}_j$, define a weighted bipartite graph ${\cal G}^{j_{\rm c}}_{ki^*}=\left(\bar {\cal V}^+_{ki^*}, \bar {\cal V}^-_{ki^*},  \bar {\cal E}_{ki^*}, W \right)$, where $\bar {\cal V}^+_{ki^*}={\cal V}^{+}_k\cup {\cal V}^{+}_{k+1}\cdots\cup {\cal V}^+_{i^*}\backslash \{x_i\}$, $\bar {\cal V}^-_{ki^*}= {\cal V}^{-}_k\cup {\cal V}^{-}_{k+1}\cdots \cup {\cal V}^-_{i^*} $, $\bar {\cal E}_{ki^*}= \{(x_i,v_l)\in {\cal E}: x_i\in \bar {\cal V}^+_{ki^*}, v_l\in \bar {\cal V}^-_{ki^*}\}$, and the weight $W(e): \bar {\cal E}_{ki^*} \rightarrow \{0,1\}$
$$W(e)=
\begin{cases} 1 & \text{if}\ e\in {\cal E}_{k}\\
0 & {\text{otherwise.}}
\end{cases}$$In other words, ${\cal G}_{ki^*}^{j_{\rm c}}$ is the subgraph of ${\cal B}(H^{j_{\rm c}}_\lambda)$ induced by vertices $\bar {\cal V}_{ki^*}^+\cup \bar {\cal V}^{-}_{ki^*}$.

\begin{proposition} \label{mainpro}
\label{nonzeroCondition} Suppose $(\bar A, \bar b)$ is structurally controllable, and there is only one nonzero entry in ${\bar F}$ with its position being $(i,j)$.  Then, for almost all $(A,b)\in {\bf CS}(\bar A,\bar b)$, there is a $(\Delta A, \Delta b)\in {\bf S}_{{\bar F}}$ such that a nonzero $n$-vector $q$ exists making $q^{\intercal}[A+\Delta A-\lambda I, b+\Delta b]=0$ for some nonzero $\lambda\in {\mathbb C}$, if and only if there exists a $k\in {\Omega}_j$ associated with which the minimum weight maximum matching of  ${\cal G}^{j_{\rm c}}_{ki^*}$ defined above is less than $|{\cal V}^{+}_k|$.%\footnote{It can be further  verified that, if $i^*\in \Omega_j$, then this condition is automatically satisfied.}
\end{proposition}

The proof relies on a series of nontrivial results on the roots of determinants of generic matrix pencils, which is postponed to the appendix.
\subsection{Necessary and Sufficient Condition}
We are now giving a necessary and sufficient condition for PTSC with general perturbation structures.

\begin{theorem}\label{NeceSuf} Consider a structurally controllable pair $(\bar A,\bar b)$ and the perturbation structure $\bar F$. For each edge $e\doteq (x_j,x_i)\in {\cal E}_{\bar F}$, let $[\bar A^e,\bar b^e]=[\bar A, \bar b]\vee \bar F^{\{e\}}$, with $\bar F^{\{e\}}$ defined in Proposition \ref{OneEdgeEquivalence}. Moreover, let ${\Omega}_j$ and ${\cal G}^{j_{\rm c}}_{ki^*}$ be defined in the same way as in Proposition \ref{nonzeroCondition}, in which $(\bar A, \bar b)$ shall be replaced with $(\bar A^e,\bar b^e)$. Then, $(\bar A,\bar b)$ is PTSC w.r.t. $\bar F$, if and only if for each edge $e\doteq (x_j,x_i)\in {\cal E}_{\bar F}$, it holds simultaneously:

 1) ${\rm grank}(\bar H[{\cal J}_n, {\cal J}_{n+1}\backslash \{j\}])=n$ or ${\rm grank}(\bar H[{\cal J}_n\backslash \{i\}, {\cal J}_{n+1}\backslash \{j\}])=n-2$, with $\bar H=[\bar A^e,\bar b^e]$;

 2) ${\Omega}_j=\emptyset$, or otherwise for each $k\in {\Omega}_j$, the minimum weight maximum matching of the bipartite ${\cal G}^{j_{\rm c}}_{ki^*}$ is~$|{\cal V}^{+}_k|$.
\end{theorem}

%be the structured system associated with the graph ${\cal G}(\bar A, \bar b)\cup {\cal G}({\bar F})- \{e\}$

\begin{proof}Follows immediately from Propositions \ref{OneEdgeEquivalence}-\ref{nonzeroCondition}.
\end{proof}

Since each step in Theorem \ref{NeceSuf} can be implemented in polynomial time, its verification has polynomial complexity. To be specific, for each edge $e\in {\cal E}_{\bar F}$, to verify Condition 1), we can invoke the Hopcroft-Karp algorithm twice, which incurs time complexity $O(n^{0.5}|{\cal E}_{[\bar A, \bar b]}\cup {\cal E}_{\bar F}|)$ $\to O(n^{2.5})$. As for Condition 2), the DM-decomposition incurs $ O(n^{2.5})$, and computing the minimum weight maximum matching of ${\cal G}_{ki^*}^{j_{\rm c}}$ costs $O(n^3)$ \cite{DB_West_graph}. Since $|\Omega_j|\le n$, for each $e\in {\cal E}_{\bar F}$, verifying Condition 2) takes at most $O(n^{2.5}+n*n^3)$. To sum up, verifying Theorem \ref{NeceSuf} incurs time complexity at most $O(|{\cal E}_{\bar F}|(n^{2.5}+n^4))$, i.e., $O(|{\cal E}_{\bar F}|n^4)$. The procedure for verifying PTSC can be summarized as follows.

{\bf Algorithm for verifying PTSC for $(\bar A, \bar b)$ w.r.t. $\bar F$}: \begin{itemize}
\item[1.] Check structural controllability of $(\bar A, \bar b)$. If yes, continue; otherwise, break and return false.
\item[2.] For each $e=(x_j,x_i)\in {\cal E}_{\bar F}$, construct $(\bar A^{e}, \bar b^{e})$, and implement the following steps:
\begin{itemize}
\item[2.1] Check whether ${\rm grank}(\bar H[{\cal J}_n, {\cal J}_{n+1}\backslash \{j\}])=n$ or ${\rm grank}(\bar H[{\cal J}_n\backslash \{i\}, {\cal J}_{n+1}\backslash \{j\}])=n-2$, with $\bar H =[\bar A^e,\bar b^e]$. If yes, continue; otherwise, return false.
\item[2.2] Construct $\Omega_j$ and ${\cal G}^{j_{\rm c}}_{ki^*}$ associated with $(\bar A^e,\bar b^e)$.
\item[2.3] For each $k\in \Omega_j$, check whether the minimum weight maximum matching of ${\cal G}^{j_{\rm c}}_{ki^*}$ is equal to $|{\cal V}^{+}_k|$. If yes, continue; otherwise, break and return false.
\end{itemize}
\item[3.] If not break, return true.
\end{itemize}

% Furthermore, it can be figured out the verification of this theorem has complexity of $O(||{\bar F}||_0n^3)$, where $||\cdot||_0$ represents the number of nonzero entries in a matrix.
\begin{example}[Example \ref{exp1} continuing]\label{exp1-con} Let us revisit Example \ref{exp1}. Consider the perturbation $[\Delta A_2, \Delta b_2]$. For edge $e=(x_5,x_4)$,  the DM-decomposition of $H_{\lambda}^{j_{\rm c}}$ ($j=5$) associated with $(\bar A^e,\bar b^e)$ and the corresponding $M_{\lambda}^j$ are respectively
$$ M_{\lambda}^{j_{\rm c}}\!=\!\left[
\begin{array}{c|c|c|c}
  s-\lambda  &    0     &     c   &    0     \\
    \hline
           & e-\lambda  &    d    &    f       \\
   \cline{2-4} \multicolumn{2}{c|}{} & -\lambda & a \\
     \cline{3-4}        \multicolumn{3}{c|}{} &  -\lambda
\end{array}
\right], M_{\lambda}^{j}\!=\!\left[\!
                              \begin{array}{c}
                                0 \\
                                r \\
                                0 \\
                                h \\
                              \end{array}
                            \!\right].
$$
%$$ M_{\lambda}^{j_{\rm c}}=\left[
%\begin{array}{c|c|c|c}
%  s-\lambda  &    0     &     c    &    0     \\
%    \hline
%           & e-\lambda  &    d    &    f       \\
%   \cline{2-4} \multicolumn{2}{c|}{} & -\lambda & a \\
%     \cline{3-4}        \multicolumn{3}{c|}{} &  -\lambda
%\end{array}
%\right], M_{\lambda}^{j}=\left[
%                              \begin{array}{c}
%                                0 \\
%                                r \\
%                                0 \\
%                                h \\
%                              \end{array}
%                            \right].
%$$
It can be obtained that, $i^*=2$, and $\Omega_j=\{1,2\}$. Since $i^*\in \Omega_j$, according to Proposition \ref{nonzeroCondition}, the corresponding perturbed system can have nonzero uncontrollable modes (in fact, if $i^*\in \Omega_j$, then the condition in Proposition \ref{nonzeroCondition} is automatically satisfied). Therefore, $(\bar A,\bar b)$ is PSSC w.r.t. $[\Delta \bar A_2,\Delta \bar b_2]$, which is consistent with Example \ref{exp1}. On the other hand, consider the perturbation $(\Delta A_1, \Delta b_1)$. For the edge $e=(x_4,x_1)$, upon letting $j=4$, we obtain ${\cal I}_j^*=\{2,3,4\}$, which means $1\notin {\cal I}_j^*$. Hence, the condition in  Proposition \ref{ZeroCondition} is satisfied. Moreover, the associated $M_{\lambda}^{j_{\rm c}}$ and $M_{\lambda}^{j}$ are respectively
$$M_{\lambda}^{j_{\rm c}}=\left[
\begin{array}{c|c|c|c}
  h  &    g     &     -\lambda   &    0     \\
    \hline
           & -\lambda  &   0    &    c       \\
   \cline{2-4} \multicolumn{2}{c|}{} & f & d \\
   \cline{3-4} \multicolumn{2}{c|}{} & a & -\lambda
\end{array}
\right], M_{\lambda}^{j}=\left[
                              \begin{array}{c}
                                l \\
                                0 \\
                                e-\lambda \\
                               0 \\
                              \end{array}
                            \right], $$from which, $i^*=1$, and $\Omega_j=\emptyset$. It means Condition 2) of Theorem \ref{NeceSuf} is satisfied.
Similar analysis could be applied to the edge $e=(x_3,x_1)$, and it turns out that both conditions in Theorem \ref{NeceSuf} hold. Therefore, $(\bar A,\bar b)$ is PTSC w.r.t. $[\Delta \bar A_1,\Delta \bar b_1]$, which is also consistent with Example \ref{exp1}. %$\hfill\square$
\end{example}

\section{Implications to SCRPs}
PTC reflects the ability of a numerical system to preserve controllability under structured perturbations. This notion is closely related to the SCRP studied in \cite{KarowStructured2009,bianchin2016observability,johnson2018structured}, where the problem is formulated as searching the smallest perturbations (in terms of the Frobenius norm or $2$-norm) with a prescribed structure that result in an uncontrollable system.

It turns out that the SCRP is feasible, if and only if the original system is PSC w.r.t. the corresponding perturbation structure (considering complex-valued perturbations). Hence, before implementing any numerical algorithms on the (single-input) SCRP, we can check whether the corresponding structured system is PTSC w.r.t. the perturbation structure. If the answer is yes and the original numerical system is controllable, then there cannot exist numerical perturbations with the prescribed structure for which the perturbed system is uncontrollable; otherwise, with probability $1$ (before looking at the exact parameters of the original system), such a structured numerical perturbation exists.

\section{Conclusion}
This paper proposes a novel notion of PTSC to study controllability preservation for a structured system under structured numerical perturbations. It is shown this notion can characterize the generic property in controllability preservation for structured systems under structured numerical perturbations. A necessary and sufficient condition is given for a single-input system to be PTSC w.r.t. a prescribed perturbation structure. Readers can refer to \cite{full-version-tac} for extensions of this work to the multi-input case.

\begin{appendix}
\subsection{Proof of Proposition \ref{mainpro}}
\begin{lemma}  \citep[Lemma 2]{Rational_function} \label{root-factor} Let $p_1(\lambda,t_1,...,t_r)$ and $p_2(\lambda,t_1,...,t_r)$ be two polynomials on the variables $\lambda,t_1,...,t_r$ with real coefficients. Then, 1) For all $(t_1,...,t_r)\in {\mathbb C}^r$, $p_1(\lambda,t_1,...,t_r)$ and $p_2(\lambda,t_1,...,t_r)$ share a common zero for $\lambda$, if and only if $p_1(\lambda,t_1,...,t_r)$ and $p_2(\lambda,t_1,...,t_r)$ share a common factor in which the leading degree for $\lambda$ is nonzero; 2) If the above-mentioned condition is not satisfied, then for almost al $(t_1,...,t_r)\in {\mathbb C}^r$ (except for a set with zero Lebsgue measure), $p_1(\lambda,t_1,...,t_r)$ and $p_2(\lambda,t_1,...,t_r)$ do not share a common zero for $\lambda$.
\end{lemma}

\begin{lemma} \label{matrix_pencil}
Let $M$ be an $n\times n$ generic matrix over the variables $t_1,...,t_r$, and $E\in \{0,1\}^{n\times n}$, where each row, as well as each column of $E$, has at most one entry being $1$ and the rest being $0$. Let $P_\lambda\doteq M-\lambda E$ be a generic matrix pencil. Moreover, ${\cal B}({P_{\lambda}})$ is the bipartite graph associated with $P_{\lambda}$ defined similarly to ${\cal B}(H_\lambda)$ by replacing $I$ with $E$ (notably, self-loops are edges in ${\cal E}_E\cap {\cal E}_M$). The following statements are true:

1) Suppose ${\cal B}({P_{\lambda}})$ contains no self-loop. Let $\gamma_{\min}({\cal B}({P_{\lambda}}))$ and $\gamma_{\max}({\cal B}({P_{\lambda}}))$ be respectively the minimum number of $\lambda$-edges and maximum number of $\lambda$-edges contained in a matching among all maximum matchings of ${\cal B}({P_{\lambda}})$. Then, the generic number of nonzero roots of $\det (P_{\lambda})$ for $\lambda$ (counting multiplicities) equals $\gamma_{\max}({\cal B}({P_{\lambda}}))- \gamma_{\min}({\cal B}({P_{\lambda}}))$.

2) If ${\cal B}({P_{\lambda}})$ is DM-reducible, then $\det (P_{\lambda})$ generically has nonzero roots for $\lambda$ whenever ${\cal B}({P_{\lambda}})$ contains a self-loop.

3) Suppose that ${\cal B}({P_{\lambda}})$ is DM-reducible. Let ${\cal T}_i$ be the subset of variables of $t_1,...,t_r$ that appear in the $i$th column of $M$. Then, for each $i\in \{1,...,n\}$, every nonzero root of $\det (P_{\lambda})$ (if exists) cannot be independent of ${\cal T}_i$.
\end{lemma}

\begin{proof}  We first prove an useful observation, that every maximum matching of ${\cal B}({P_{\lambda}})$ corresponds to a nonzero term (monomial) in $\det (P_{\lambda})$ that cannot be zeroed out by other terms, which is fundamental to the following proofs. For this purpose, consider a term $\lambda^{r_1}t_{l_1}t_{l_2}\cdots t_{l_{r_2}}$ associated with a maximum matching of ${\cal B}({P_{\lambda}})$, where $r_1+r_2=n$ and $\{l_1,...,l_{r_2}\}\subseteq \{1,...,r\}$. The only case to zero out $\lambda^{r_1} t_{l_1}t_{l_2}\cdots t_{l_{r_2}}$ in $\det (P_{\lambda})$ is that there exists a term being $\lambda^{r_1} t_{l_1}t_{l_2}\cdots t_{l_{r_2}}$ associated with another maximum matching of ${\cal B}({P_{\lambda}})$.  Now suppose that $R$ (resp. $C$) is the set of row indices (resp. column indices) of $t_{l_1},...,t_{l_{r_2}}$ in $M$, recalling that each $t_i$ ($i=1,...,r$) appears only once. Then, $\lambda^{r_1}$ must correspond to $r_1$ `1' entries in $E[{\cal J}_n\backslash R, {\cal J}_n\backslash C]$. However, as each row and each column has at most one `$1$' in $E$, the aforementioned configuration for `$1$' entries is unique, a contraction to the existence of two different maximum matchings associated with $\lambda^{r_1} t_{l_1}t_{l_2}\cdots t_{l_{r_2}}$.

We now prove 1). Suppose $\det (P_{\lambda})$ can be factored as $\lambda^{l}f(\lambda,t_1,...,t_r)$, where $l\in {\mathbb N}$ and $f(\lambda,t_1,...,t_r)$ is a polynomial of $\lambda,t_1,...,t_r$ that does not contain factors in the form of $\lambda^{\bar l}$ for any $\bar l\ge 1$. Because of the above observation, every term associated with a maximum matching of ${\cal B}({P_{\lambda}})$ must contain the factor $\lambda^l$. Therefore, the number $l$ of zero roots of $\det (P_{\lambda})$ equals  $\gamma_{\min}({\cal B}({P_{\lambda}}))$.  In addition, the maximum degree of $\lambda$ in $\lambda^{l}f(\lambda,t_1,...,t_r)$ appears in a term associated with a maximum matching containing the maximum number of $\lambda$-edges, which is exactly $\gamma_{\max}({\cal B}({P_{\lambda}}))$. The conclusion in 1) then follows immediately from the fundamental theorem of algebra.

Next, we prove 2). Consider a self-loop with the entry being $t_l-\lambda$ ($1\le l \le r$). As ${\cal B}({P_{\lambda}})$ is DM-reducible, every nonzero entry must be contained in $\det (P_{\lambda})$ by Definition \ref{DM-def}, which means $t_l-\lambda$ is contained in some term $(t_l-\lambda)f$ of $\det (P_{\lambda})$, where $f$ denotes a polynomial over variables $\{t_1,...,t_r\}\backslash \{t_l\}$ and $\lambda$. This term can be written as the sum of two terms $t_lf$ and $t_l\lambda f$, which indicates $\det (P_{\lambda})$ contains at least two monomials whose degrees for $\lambda$ differ from each other. Then, following the similar reasoning to the proof of 1),  $\det (P_{\lambda})$ contains at least one nonzero root.

We are now proving 3). Suppose such a nonzero root exists that is independent of ${\cal T}_i$ for some $i\in \{1,...,n\}$, and denote it by $z$. Let ${\cal T}[{\cal I}_1,{\cal I}_2]$ be the set of variables in $t_1,...,t_r$ that appear in $M[{\cal I}_1,{\cal I}_2]$ for ${\cal I}_1,{\cal I}_2\subseteq {\cal J}_n$, and let $R({\cal T}_s)$ (resp. $C({\cal T}_s)$) be the set of row (resp. column) indices of variables ${\cal T}_s\subseteq \{t_1,...,t_r\}$.  Suppose $[P_{\lambda}]_{k_0,i}=\lambda$ for some $k_0\in {\cal J}_n\backslash R({\cal T}_i)$ ($k_0$ can be empty). Upon letting all $t_k\in {\cal T}_i$ be zero, we obtain ($P_z=M-z E$)
$${\small\begin{array}{c}\begin{aligned} &\det (P_z)=\sum\nolimits_{j=1}^n(-1)^{i+j}[P_{z}]_{ji}\det(P_z[{\cal J}_n\backslash \{j\}, {\cal J}_n\backslash \{i\}])\\
&=z\cdot \det (P_z[{\cal J}_n\backslash \{k_0\}, {\cal J}_n\backslash \{i\}])=0, \end{aligned}\end{array}}$$
which indicates
\begin{equation}\label{lambda_zero} \det (P_z[{\cal J}_n\backslash \{k_0\}, {\cal J}_n\backslash \{i\}])=0,\end{equation}
as $z\ne 0$. Since $(P_{\lambda}[{\cal J}_n\backslash \{k_0\}, {\cal J}_n\backslash \{i\}])$ has full generic rank from Lemma \ref{reduciblility}, it concludes that $z$ depends solely on the variables ${\cal T}[{\cal J}_n\backslash \{k_0\}, {\cal J}_n\backslash \{i\}]$. Similarly, because of (\ref{lambda_zero}), for each $t_k\in {\cal T}_i$, fixing all $t_j\in {\cal T}_i\backslash \{t_k\}$ to be zero yields
$$\det (P_z[{\cal J}_n\backslash R(\{t_k\}), {\cal J}_n\backslash \{i\}])=0,$$
which indicates that $z$ depends on the variables ${\cal T}[{\cal J}_n\backslash R(\{t_k\}), {\cal J}_n\backslash \{i\}]$, being independent of the remaining variables. Taking the intersection of ${\cal T}[{\cal J}_n\backslash \{j\}, {\cal J}_n\backslash \{i\}]$ over all $j\in R({\cal T}_i)\cup \{k_0\}$, we obtain ${\cal T}[\Theta,{\cal J}_n\backslash \{i\}]$, where $\Theta\doteq {\cal J}_n\backslash (R({\cal T}_i)\cup \{k_0\})$. That is, $z$ depends on variables ${\cal T}[\Theta,{\cal J}_n\backslash \{i\}]$, and makes $P_{\lambda}[\Theta,{\cal J}_n\backslash \{i\}]$ row rank deficient. However, for each pair $(j,l)$, $j\! \in\! {\cal J}_n\backslash \{i\}$, $l\in  R({\cal T}_i)$, it~holds
$$\begin{array}{c}\begin{aligned} &{\rm grank} (P_{\lambda}[\Theta,{\cal J}_n\backslash \{i,j\}])\\
&\mge {\rm grank}((M\!-\!\lambda E)[{\cal J}_n\backslash \{l\} ,{\cal J}_n\backslash \{j\}])\!-\!(|R({\cal T}_i)\cup\{k_0\}|\!-1\!)\\
&\myeq n-|R({\cal T}_i)\cup\{k_0\}|=|\Theta|,\end{aligned}\end{array}$$
where (a) is due to that $P_{\lambda}[\Theta,{\cal J}_n\backslash \{i,j\}]$ is obtained by deleting $|R({\cal T}_i)\cup\{k_0\}|\!-1\!$ rows from $P_\lambda[{\cal J}_n\backslash \{l\} ,{\cal J}_n\backslash \{j\}]$, and (b) comes from ${\rm grank}(P_{\lambda}[{\cal J}_n\backslash \{l\} ,{\cal J}_n\backslash \{j\}])=n-1$ by 2) of Lemma \ref{reduciblility}. That is, after deleting any column from $P_{\lambda}[\Theta,{\cal J}_n\backslash \{i\}]$, the resulting matrix remains of full row generic rank, which induces at least one nonzero polynomial equation constraint on $z$ and ${\cal T}[\Theta,{\cal J}_n\backslash \{i,j\}]$. This indicates $z$ depends on ${\cal T}[\Theta,{\cal J}_n\backslash \{i,j\}]$, or equivalently, being independent of ${\cal T}[\Theta,\{j\}]$, for each $j\in {\cal J}_n\backslash \{i\}$. It finally concludes that $z$ is independent of the variables ${\cal T}[\Theta, {\cal J}_n\backslash \{i\}]$, causing a contraction. Therefore, the assumed $z$ cannot exist.
\end{proof}
\vspace{-0.1cm}
\begin{lemma}\label{immediate} Let $M^{j_{\rm c}}_\lambda$ and ${\Omega_j}$ be defined in (\ref{DMdecomposition}) and (\ref{important-set}). For each $k\in {\Omega_j}$, let $\tilde{M}^{j_{\rm c}}_{ki^*}(\lambda)\doteq M_{\lambda}^{j_{\rm c}}[\bar {\cal V}^+_{ki^*},\bar {\cal V}^-_{ki^*}]$, i.e.,
{\small$$
 \tilde{M}^{j_{\rm c}}_{ki^*}(\lambda)\!=\!\left[\!
                               \begin{array}{ccc}
                                 M_k^{j_{\rm c}}(\lambda) & \cdots & M_{ki^*}^{j_{\rm c}}(\lambda) \\
                                 0 & \ddots & \vdots \\
                                 0 & 0 & M^{j_{\rm c}}_{i^*}(\lambda)[{\cal J}_{|{\cal V}^+_{i^*}|}\backslash \{\bar i\},{\cal J}_{|{\cal V}^+_{i^*}|}]
                               \end{array}
                             \!\right].$$}Then, $\tilde{M}^{j_{\rm c}}_{ki^*}(\lambda)$ generically has full row rank when $\lambda \in \{z\in {\mathbb C}\backslash \{0\}: \det M_k^{j_{\rm c}}(z)=0 \}$, if and only if the minimum weight maximum matching of the bipartite ${\cal G}^{j_{\rm c}}_{ki^*}$ is less than $|{\cal V}^{+}_k|$.
\end{lemma}

% \left[
%                                \begin{array}{cc}
%                                  M_k^{j_{\rm c}}(\lambda) & M^{j_{\rm c}}_{ki^*}(\lambda) \\
%                                  0 & M^j_{i^*}(\lambda)[{\cal J}_{|{\cal V}^+_{i^*}|}\backslash \{\bar i\},{\cal J}_{|{\cal V}^+_{i^*}|}] \\
%                                \end{array}
%                              \right].
%\footnote{If $k=i^*$, then $\tilde{M}^{j_{\rm c}}_{ki^*}(\lambda)=M^j_{i^*}(\lambda)[{\cal J}_{|{\cal V}^+_{i^*}|}\backslash \{\bar i\},{\cal J}_{|{\cal V}^+_{i^*}|}].$} that if the $\bar k$th column contains only one nonzero entry being $\lambda$ when $|{\cal V}^+_k|>1$, then ${\cal G}^{j_{\rm c}}_k$ is not DM-reducible. T

\begin{proof}  Sufficiency: Let $n_1=|{\cal V}^+_k \cup {\cal V}^+_{k+1}\cdots \cup{\cal V}^+_{i^*}|$. By Lemma \ref{reduciblility}, ${\rm mt}({\cal G}^{j_c}_{i^*}-\{x_i\})=|{\cal V}^+_{i^*}|-1$. Consequently, ${\cal G}^{j_{\rm c}}_{ki^*}$ has a maximum matching with size $n_1-1$ from its structure. Suppose that ${\cal G}^{j_{\rm c}}_{ki^*}$ has a maximum matching with weight less than $|{\cal V}^{+}_k|$. Then, ${\cal B}(\tilde{M}^{j_{\rm c}}_{ki^*}(\lambda)[{\cal J}_{n_1-1},{\cal J}_{n_1}\backslash \{\bar k\}])$ for some $\bar k\in \{1,...,|{\cal V}^-_k|\}$ must have a matching with size $n_1-1$. Indeed, if this is not true, then any maximum matching of ${\cal G}^{j_{\rm c}}_{ki^*}$ must matches ${\cal V}^-_k$, which certainly leads to a weight equaling $|{\cal V}^+_k|$, noting that each edge not incident to ${\cal V}^-_k$ has a zero weight. Furthermore, due to the DM-irreducibility of ${\cal G}^{j_{\rm c}}_k$, from Lemma \ref{matrix_pencil}, any nonzero root of $\det M_k^{j_{\rm c}}(\lambda)$ cannot be independent of the variables in the $\bar k$th column of $M_k^{j_{\rm c}}(\lambda)$.\footnote{Note the case where $|{\cal V}^+_k|=1$ and $M_k^{j_{\rm c}}(\lambda)=\lambda$ has been excluded by the nonzero root assumption.} Therefore, $\det M_k^{j_{\rm c}}(\lambda)$ and $\det \tilde{M}^{j_{\rm c}}_{ki^*}(\lambda)[{\cal J}_{n_1-1},{\cal J}_{n_1}\backslash \{\bar k\}]$ generically do not share a common nonzero root, since the latter determinant cannot contain the variables in the $\bar k$th column of $M_k^{j_{\rm c}}(\lambda)$ (except for $\lambda$).  That is, $\det \tilde{M}^{j_{\rm c}}_{ki^*}(\lambda)[{\cal J}_{n_1-1},{\cal J}_{n_1}\backslash \{\bar k\}]$ is generically nonzero for $\lambda \in \{z\in {\mathbb C}\backslash \{0\}: \det M_k^{j_{\rm c}}(z)=0\}$, leading to the full row rank of $\tilde{M}^{j_{\rm c}}_{ki^*}(\lambda)$.

Necessity: If $k=i^*$, the necessity is obvious. Consider $k<i^*$. Suppose the minimum weight maximum matching of ${\cal G}^{j_{\rm c}}_{ki^*}$ equals $|{\cal V}^+_k|$. Then,  based on the above analysis, any maximum matching of ${\cal G}^{j_{\rm c}}_{ki^*}$ must match ${\cal V}^-_k$, which leads to a zero determinant of ${\tilde M}^{j_{\rm c}}_{ki^*}(\lambda)$, due to the block-triangular structure of ${\tilde M}^{j_{\rm c}}_{ki^*}(\lambda)$ and the fact that $\det M_k^{j_{\rm c}}(\lambda)=0$, contradicting the full row rank of ${\tilde M}^{j_{\rm c}}_{ki^*}(\lambda)$.
\end{proof}

{\bf Proof of Proposition \ref{mainpro}}: In the following, suppose
vertex $x_i$ corresponds to the $\hat i$th row of $M^{j_{\rm c}}_\lambda$ after the permutation by $P$, i.e. $[M_\lambda^j]_{\hat i}=[PH_\lambda^j]_{\hat i}=[H_\lambda^j]_i$. Recall the involved $H_\lambda, M^{j_c}_{\lambda}$ and their submatrices are treated as generic matrix pencils.

 Sufficiency: From Lemma \ref{matrix_pencil}, we know that for each $k\in \Omega_j$, there exists a nonzero $\lambda$ making $\det M_k^{j_{\rm c}}(\lambda)=0$.  To distinguish such value from the variable $\lambda$, we denote it by $z$ (i.e., $z\ne 0$ and $\det M_k^{j_{\rm c}}(z)=0$). From Lemma \ref{root-factor}, it generically holds that $\det M_l^j(z)\ne 0$ for all $l\in \{1,...,b\}\backslash \{k\}$, as $\det M_k^{j_{\rm c}}(\lambda)$ and $\det M_l^j(\lambda)$ do not share any common factor except the factor $\lambda$. Due to the block-triangular structure of $M^{j_{\rm c}}_z$ (obtained by replacing $\lambda$ with $z$ in $M^{j_{\rm c}}_\lambda$), it can be seen readily that if $\tilde{M}^{j_{\rm c}}_{ki^*}(z)$ defined in Lemma \ref{immediate} generically has full row rank, then $M^{j_{\rm c}}_z[{\cal J}_n\backslash \{\hat i\}, {\cal J}_{n}]$ will do. The former condition has been proven in Lemma \ref{immediate}.

Note also that $M^{j_{\rm c}}_z$ generically has rank $n-1$ as otherwise $[\bar A-zI,\bar b]$ generically has rank less than $n$, contradicting the structural controllability of $(\bar A, \bar b)$. Therefore, from Lemma \ref{basiclemma}, letting ${\hat q}$ be a nonzero vector in the left null space of $M^{j_{\rm c}}_z$, we have ${\hat q}_{\hat i}\ne 0$. For almost all $(A,b)\in {\bf CS}(\bar A, \bar b)$, by letting $[\Delta A, \Delta b]_{ij}=-1/{\hat q}_{\hat i} \sum \nolimits_{l=1}^n{\hat q}_l(P[A-zI,b])_{lj},$ we get
$$\begin{array}{c}\begin{aligned}&{\hat q}^{\intercal}P([A-zI,b]+[\Delta A, \Delta b])[{\cal J}_n,\{j\}]\\
&={\hat q}_{\hat i}(P[\Delta A, \Delta b])_{\hat i,j}+\sum \nolimits_{l=1}^n {\hat q}_l(P[A-zI,b])_{lj}\\
&=0,\end{aligned} \end{array}$$
where the second equality is due to $(P[\Delta A,\Delta b])_{\hat i,j}=[\Delta A, \Delta b]_{ij}$. Upon defining $q^{\intercal}\doteq {\hat q}^{\intercal}P$, we have
$$q^{\intercal}([A-zI,b]+[\Delta A, \Delta b])=0,$$
which comes from the fact  $q^{\intercal}H_z^{j_{\rm c}}Q=0$ and $Q$ is invertible.

Necessity: For the existence of $q$ making the condition in Proposition \ref{mainpro} satisfied, it is necessary $H_\lambda^{j_{\rm c}}$ should be of rank deficient at some nonzero value for $\lambda$ (generically). Denote such a value by $z$ for the sake of distinguishing it from the variable $\lambda$. Since DM-decomposition does not alter the rank, $M_{z}^{j_{\rm c}}$ should be of row rank deficient too. From the block-triangular structure of $M_z^{j_{\rm c}}$ (see (\ref{DMdecomposition})), there must exist some $k\in \{1,...,b\}$, such that $M^{j_{\rm c}}_{k}(z)$ is singular generically. From Lemma \ref{matrix_pencil}, such an integer $k$ must correspond to a ${\cal G}_k^{j_{\rm c}}$ satisfying ${\gamma}_{\rm nz}({\cal G}_k^{j_{\rm c}})=1$. We consider two cases: i) $k> i^*$, and ii) $k\le i^*$.

In case i), since $k>i^*$, from the upper block-triangular structure of $M_z^{j_{\rm c}}$, it is clear that $M_z^{j_{\rm c}}[{\cal J}_n\backslash \{\hat i\}, {\cal J}_n]$ is of row rank deficient when $\det M^{j_{\rm c}}_k(z)=0$. Note that ${\rm grank}(M_{z}^{j_{\rm c}})\ge n-1$ as otherwise ${\rm grank}(H_z)< n$, which is contradictory to the structural controllability of $(\bar A, \bar b)$. Consequently, $M_{z}^{j_{\rm c}}$ has a left null space with dimension one. Denote by $\hat q$ the vector spanning that space. From Lemma \ref{basiclemma}, ${\hat q}_{\hat i}=0$. As a result, for any $[\Delta A, \Delta  b]\in {\bf S}_{[ {\Delta \bar A},  {\Delta \bar b}]}$,
$$\hat q^{\intercal} \big\{M^j_{z}+(P[\Delta A, \Delta  b])[{\cal J}_n,\{j\}]\big\}\myeqa \hat q^{\intercal} M^j_z \ne 0,$$
where (a) results from $(P[\Delta A, \Delta b])_{\hat ij}=[\Delta A, \Delta b]_{ij}$, and the inequality  from $\hat q^{\intercal} [M^{j_{\rm c}}_z,M^j_z]\ne 0$, as otherwise $\hat q^{\intercal} P H_z=0$ meaning that $z$ will be an uncontrollable mode (noting that $\hat q^{\intercal}  [M^{j_{\rm c}}_z,M^j_z]=\hat q^{\intercal} P H_z Q=0$, and $Q$ is invertible). Consequently, case i) cannot lead to the required results.

Therefore, $k$ must fall into case ii). Now suppose that the minimum weight maximum matching of ${\cal G}_{ki^*}^{j_{\rm c}}$ is equal to $|{\cal V}_k^+|$. Then, from Lemma \ref{immediate} and by the block-triangular structure of $M_z^{j_{\rm c}}$, we obtain that $M_z^{j_{\rm c}}[{\cal J}_n\backslash \{\hat i\}, {\cal J}_n]$ is generically of row rank deficient. Following the similar reasoning to case i), it turns out that the requirement in Proposition \ref{mainpro} cannot be satisfied. This proves the necessity. $\hfill\square$
\end{appendix}

%\end{theorem}

{\bibliographystyle{unsrt}
{\footnotesize
\bibliography{yuanz3}
}}

\end{document}